\documentclass[11pt]{article}
\usepackage{amsmath, amssymb, amsthm, mathtools}
\usepackage{hyperref}
\usepackage{natbib}
\usepackage[margin=1in]{geometry}

\newtheorem{theorem}{Theorem}[section]
\newtheorem{lemma}[theorem]{Lemma}

\newtheorem{corollary}[theorem]{Corollary}

\theoremstyle{definition}

\theoremstyle{remark}

\DeclareMathOperator{\OPT}{OPT}
\DeclareMathOperator{\SREV}{SREV}
\DeclareMathOperator{\BREV}{BREV}
\DeclareMathOperator{\SRev}{SRev}

\allowdisplaybreaks

\begin{document}
\title{Improved Revenue Guarantees for Selling Separately and Bundling \thanks{The main result was entirely obtained by Cogentic, an agentic framework for mathematical discovery, using an interval version of Gemini as the base model. The authors contextualized the findings and verified the proofs. The full exposition here is due to the authors aided by different AI models.

The following authors have additional affiliations beyond Google Research: Yang Cai (Yale University) and Vineet Gupta (Google DeepMind).}}
\author{
\parbox{0.95\textwidth}{\centering
Yang Cai,
Vineet Gupta,
Yanchen Jiang,
Christopher Liaw,\\[0.3em]
Aranyak Mehta,
Grigoris Velegkas,
Di~Wang
}\\[1.2em]
\normalsize Google Research\\[0.3em]
\small\texttt{\{caiy, vineet, yanchenjiang, cvliaw, aranyak, gvelegkas, wadi\}@google.com}
}
\date{\today}
\maketitle

\begin{abstract}
We study how much revenue a seller can lose by restricting attention to selling separately or grand bundling, in the setting of a single additive buyer with independent item values. Although revenue-optimal mechanisms can require lotteries and infinite menus, Babaioff, Immorlica, Lucier, and Weinberg showed that the better of these two simple formats always achieves a constant fraction of optimal revenue. We prove that
\[
    \OPT \le 3.52\max\{\SREV,\BREV\},
\]
where $\SREV$ and $\BREV$ are the optimal revenues from selling separately and grand bundling, respectively. This improves the previous best-known approximation factor of $5.2$ due to Ma and Simchi-Levi and narrows the gap to the known lower bound of $2$.
\end{abstract}

\section{Introduction}
How much revenue must a seller forgo to use a simple selling rule? For a single item and a single buyer, the answer is zero: an appropriately chosen take-it-or-leave-it price maximizes expected revenue~\cite{myerson1981}. With multiple items, even this basic selling problem changes fundamentally. A seller facing one additive buyer whose item values are independent may benefit from randomized allocations, and the optimal mechanism can require an infinite menu of lotteries~\cite{daskalakis2017}. Computing an exactly optimal mechanism is also computationally intractable in basic instances of this setting~\cite{daskalakis2014}. Thus, the difficulty of optimal selling arises even without competition among buyers, complementarities between items, or correlation between their values.

This contrast motivates the study of \emph{simple versus optimal mechanisms}: rather than characterizing the optimal mechanism, can a seller obtain a substantial fraction of its revenue using transparent, easily implemented selling rules? Approximation guarantees provide a way to quantify the tradeoff between simplicity and optimality. 

For an additive buyer, two particularly natural formats are \emph{separate selling}, which posts a price for each item and lets the buyer purchase any subset, and \emph{grand bundling}, which offers all items together at a single price. Write $\SREV$ and $\BREV$ for their respective optimal expected revenues. These formats use different ways to extract revenue: separate prices can accommodate heterogeneous item distributions, while bundling can exploit the aggregation of independent values. Neither format alone guarantees a constant fraction of optimal revenue for arbitrary independent item values~\cite{hart2012,babaioff2014}. The question is whether choosing between them is enough.

\paragraph{From logarithmic to constant-factor guarantees.}
Hart and Nisan~\cite{hart2012} initiated the study of approximation guarantees for these simple mechanisms, proving that separate selling achieves an $O(\log^2 n)$-approximation for $n$ independent items. Li and Yao~\cite{li2013} improved this bound to $O(\log n)$ and showed that grand bundling achieves a constant-factor approximation when the item values are also identically distributed. Babaioff, Immorlica, Lucier, and Weinberg~\cite{babaioff2014} then proved that, for arbitrary independent item values,
\[
    \OPT \le 6\max\{\SREV,\BREV\},
\]
where $\OPT$ is the optimal expected revenue over all incentive-compatible and individually rational mechanisms, including randomized ones. Their theorem established that choosing the better of two elementary deterministic formats suffices for a constant-factor guarantee, even though either format can perform poorly on its own. 

The simple-versus-optimal program has since reached much broader settings. Under the item-independence assumption, Cai and Zhao~\cite{cai2017} established constant-factor guarantees for multiple XOS buyers using sequential posted-price mechanisms with or without entry fees. Cai, Oikonomou, and Zhao~\cite{cai2022stoc} developed polynomial-time algorithms for computing approximately optimal simple mechanisms in this setting. Cai, Chen, and Wu~\cite{cai2023} established constant-factor equilibrium-revenue guarantees for suitably augmented simultaneous auctions with subadditive bidders. These developments demonstrate the scope of settings where simple mechanisms are provably competitive with respect to the optimal but complex one.

\paragraph{The quantitative gap.}
In the original single additive buyer setting, Ma and Simchi-Levi~\cite{ma2021}  improved the approximation factor to $5.2$. Their work first circulated in 2015 and was published in 2021. On the other hand, Rubinstein~\cite{rubinstein2016} constructed independent-item instances in which a mechanism that partitions the items into disjoint bundles earns $(2-o(1))\max\{\SREV,\BREV\}$. Consequently, no approximation factor below $2$ is possible for the better of separate selling and grand bundling. The gap between $2$ and $5.2$ leaves substantial uncertainty about the revenue cost of restricting attention to these two formats.

\paragraph{Our result.}
We substantially narrow this gap by proving the following guarantee.
\begin{quote}
\textbf{Main result (Theorem~\ref{thm:main}).}
\emph{For a single additive buyer with independent nonnegative item values,}
\[
    \OPT \le 3.52\max\{\SREV,\BREV\}.
\]
\end{quote}
The bound holds uniformly over the number of items and their value distributions; the distributions need not be identical or satisfy specific properties, e.g., MHR, or regularity. 

\subsection{Technical Overview}
Our analysis builds on the duality framework of Cai, Devanur, and Weinberg~\cite{cai2016}. Let $v_j$ denote the buyer's value for item $j$. We define the total welfare as $V=\sum_j v_j$, the maximum item value as $M_{\max}=\max_j v_j$, and the joint revenue scale as $M=\max\{\SREV,\BREV\}$. Their framework yields the upper bound
\[
    \OPT \le \SREV+\mathbb{E}[V-M_{\max}].
\]
The remaining task is to bound the expected \emph{non-favorite welfare}, $\mathbb{E}[V-M_{\max}]$, which represents the expected total value of all items except the buyer's favorite. We accomplish this through a capped-welfare reduction, followed by a comparison of two extremal second-moment bounds.

\paragraph{A direct reduction to capped welfare.}
We define the capped welfare as
\[
    V_M=\sum_{j=1}^n\min\{v_j,M\},
\]
and show that the expected non-favorite welfare is at most the expectation of this capped sum:
\[
    \mathbb{E}[V-M_{\max}]\le\mathbb{E}[V_M].
\]
While prior analyses typically bound the non-favorite welfare by truncating item values at $\SREV$, we depart from this paradigm by truncating at the joint scale $M$. Traditionally, when capped at $\SREV$, the expectation $\mathbb{E}[V - M_{\max}]$ is decomposed into a $\textsc{Core}$ (the capped welfare) and a $\textsc{Tail}$ (the remainder)~\citep{cai2016}. By factoring the bundle revenue into the truncation threshold, we can focus entirely on a $\textsc{Core}$-like quantity and bypass the $\textsc{Tail}$ analysis completely. Finally, we bound $\mathbb{E}[V_M]$ by comparing upper and lower bounds on its second moment, $\mathbb{E}[V_M^2]$, establishing that this expected truncated welfare is at most $2.5168 M$.

\subsection{Further Related Work}
\paragraph{Core--tail decompositions and duality.}
The core--tail decomposition introduced by Li and Yao~\cite{li2013} separates bounded values from rare, high values and is central to the analyses of Babaioff et al.~\cite{babaioff2014} and Ma and Simchi-Levi~\cite{ma2021}. The latter exploit the fact that the worst cases for their core and tail bounds cannot occur simultaneously, and their $5.2$ guarantee also accommodates production costs; the revenue guarantee discussed here is the zero-cost special case. Separately, the duality framework of Cai, Devanur, and Weinberg~\cite{cai2016} recovers the $6$-approximation. Cai, Devanur, Goldner, and McAfee~\cite{cai2019} refine that analysis to show
\[
    \OPT\le(1+a)\SREV+
    \left(2+\frac{2}{a^2}\right)\BREV
    \qquad\text{for every }a>0,
\]
which yields a $5.382$-approximation. Although numerically weaker than the $5.2$ bound, this duality-based analysis is the closest technical antecedent of our approach.

\paragraph{Other settings.}
Early work established revenue guarantees for unit-demand buyers~\cite{chawla2007,chawla2010}. Beyond the additive single-buyer setting, simple-mechanism guarantees have been established for a single subadditive buyer~\cite{rubinstein2015}, proportional complementarities~\cite{cai2019}, and dependent item values with guarantees parameterized by the strength of dependence~\cite{cai2021}. For multiple subadditive buyers, D\"utting, Kesselheim, and Lucier~\cite{dutting2020} obtained an $O(\log\log n)$-approximation using sequential two-part tariffs, where $n$ denotes the number of items. The constant-factor equilibrium-revenue guarantees for simultaneous auctions with subadditive bidders~\cite{cai2023} extend earlier guarantees for additive bidders~\cite{daskalakis2022}.

\section{Preliminaries}

We consider a single buyer with additive valuations over $n$ independently distributed items. Let $v_j \ge 0$ denote the buyer's private value for item $j \in \{1, \dots, n\}$. We define:
\begin{itemize}
    \item $\OPT$: the optimal expected revenue of any randomized Bayesian Incentive Compatible (BIC) and Ex-Post Individually Rational (IR) mechanism.
    \item $V = \sum_{j=1}^n v_j$ and $M_{\max} = \max_{1 \le j \le n} v_j$: the grand bundle value and the maximum single-item value.
    \item $\SRev_j = \sup_{p \ge 0} p \Pr(v_j > p)$ and $\SREV = \sum_{j=1}^n \SRev_j$: the optimal single-item revenue for item $j$ and the total revenue from selling items separately.
    \item $\BREV = \sup_{p \ge 0} p \Pr(V > p)$ and $M = \max\{\SREV, \BREV\}$: the optimal grand-bundle revenue and the benchmark revenue. Note that $\SREV \le M$ and $\BREV \le M$.
    \item $p_j(z) = \Pr(v_j > z)$ and $K(z) = \sum_{j=1}^n p_j(z)$: the survival function of item $j$ and the aggregate survival function. Since $z p_j(z) \le \SRev_j$ for all $z > 0$, summing over $j$ gives
    \[
        z K(z) \le \SREV \le M \implies K(z) \le \frac{M}{z} \quad \text{for all } z > 0.
    \]
\end{itemize}

\section{Proof Overview.}
 Our proof proceeds in four transparent steps:
\begin{enumerate}
    \item \textbf{Duality Decomposition (Lemma~\ref{lem:cdw}):} By the duality framework of Cai, Devanur, and Weinberg~\cite{cai2016}, any Bayesian Incentive Compatible (BIC) and Ex-Post Individually Rational (IR) mechanism satisfies
    \[
        \OPT \le \SREV + \mathbb{E}[V - M_{\max}].
    \]
    \item \textbf{Survival Integral Bound (Lemma~\ref{lem:survival}):} Expressing $V - M_{\max}$ in terms of the item survival probabilities $p_j(z) = \Pr(v_j > z)$ and their sum $K(z) = \sum_{j=1}^n p_j(z)$, item independence implies $\Pr(M_{\max} > z) \ge 1 - e^{-K(z)}$, which gives
    \[
        \mathbb{E}[V - M_{\max}] \le \int_0^\infty \bigl(K(z) - 1 + e^{-K(z)}\bigr)\,dz.
    \]
    \item \textbf{Reduction to the Truncated Sum (Lemma~\ref{lem:truncate}):} Because single-item pricing ensures $K(z) \le \SREV/z \le M/z$, the term $1 - e^{-K(z)}$ subtracted on $[0, M]$ is at least as large as the entire tail of the integral on $[M, \infty)$. Consequently,
    \[
        \int_0^\infty \bigl(K(z) - 1 + e^{-K(z)}\bigr)\,dz \le \int_0^M K(z)\,dz = \mathbb{E}[V_M],
    \]
    where $V_M = \sum_{j=1}^n \min(v_j, M)$ is the sum of item values truncated at $M$.
    \item \textbf{Second-Moment Tradeoff for $V_M$ (Section~\ref{sec:second-moment}):} To bound $E_M = \mathbb{E}[V_M]$, we analyze the second moment $S_2 = \mathbb{E}[V_M^2]$ from two directions:
    \begin{itemize}
        \item \emph{Upper bound (Lemma~\ref{lem:s2-upper}):} Since the truncated values $V_j = \min(v_j, M)$ are independent, their variances add. Subject to the single-item tail cap $\Pr(V_j > z) \le \SRev_j/z$, $\mathbb{E}[V_j^2]$ is maximized by pushing probability mass as far right as monotonicity allows. Aggregating across items via Jensen's inequality gives $S_2 \le E_M^2 + M^2 y(E_M/M)$, where $y(x) = 2 - e^{1-x}$ for $x > 1$.
        \item \emph{Lower bound (Lemma~\ref{lem:em-bound}):} Subject to the grand-bundle tail cap $\Pr(V_M > z) \le M/z$, a distribution with mean $E_M$ minimizes $S_2$ by packing probability mass as far left as possible, yielding $E_M \le M\bigl(1 + \ln\frac{S_2 + M^2}{2M^2}\bigr)$.
    \end{itemize}
    Combining the upper and lower bounds on $S_2$ yields a single-variable inequality in $x = E_M/M$, forcing $E_M \le 2.5168 M$ and hence $\OPT \le \SREV + 2.5168 M \le 3.5168 M$.
\end{enumerate}

\section{Reducing $\OPT$ to the Truncated Sum}\label{sec:reduction}

Our starting point is the standard decomposition of Cai, Devanur, and Weinberg~\cite{cai2016}, which bounds the optimal revenue by the sum of separate item revenues plus the expected value of all non-favorite items.

\begin{lemma}[Cai, Devanur, and Weinberg~\cite{cai2016}]\label{lem:cdw}
The expected revenue of any BIC and Ex-Post IR mechanism satisfies
\[
    \OPT \le \SREV + \mathbb{E}[V - M_{\max}].
\]
\end{lemma}

Next, we express $\mathbb{E}[V - M_{\max}]$ in terms of the aggregate survival function $K(z)$.

\begin{lemma}\label{lem:survival}
$\mathbb{E}[V - M_{\max}] \le \int_0^\infty \bigl(K(z) - 1 + e^{-K(z)}\bigr)\,dz$.
\end{lemma}
\begin{proof}
For any realization $v = (v_1, \dots, v_n) \in \mathbb{R}_{\ge 0}^n$, we can write $v_j = \int_0^\infty \mathbf{1}[v_j > z]\,dz$ and $M_{\max} = \int_0^\infty \mathbf{1}[M_{\max} > z]\,dz$. Thus,
\[
    V - M_{\max} = \int_0^\infty \left(\sum_{j=1}^n \mathbf{1}[v_j > z] - \mathbf{1}[M_{\max} > z]\right) dz.
\]
Notice that if $M_{\max} > z$, then at least one $v_j > z$, so the integrand $\sum_{j=1}^n \mathbf{1}[v_j > z] - \mathbf{1}[M_{\max} > z]$ is pointwise non-negative for every $z \ge 0$. Applying Tonelli's theorem to exchange expectation and integration yields
\[
    \mathbb{E}[V - M_{\max}] = \int_0^\infty \bigl(K(z) - \Pr(M_{\max} > z)\bigr)\,dz.
\]
By independence of the items and the inequality $1 - x \le e^{-x}$, we have
\[
    \Pr(M_{\max} \le z) = \prod_{j=1}^n \bigl(1 - p_j(z)\bigr) \le \prod_{j=1}^n e^{-p_j(z)} = e^{-K(z)},
\]
so $\Pr(M_{\max} > z) \ge 1 - e^{-K(z)}$. Substituting this into the integral completes the proof.
\end{proof}

We now show that the improper integral in Lemma~\ref{lem:survival} is bounded above by the integral of $K(z)$ truncated at $M$. Intuitively, the negative correction $-(1 - e^{-K(z)})$ on the interval $[0, M]$ more than compensates for the entire tail of the integral on $[M, \infty)$.

\begin{lemma}\label{lem:truncate}
For any non-increasing function $K(z) \ge 0$ satisfying $K(z) \le M/z$,
\[
    \int_0^\infty \bigl(K(z) - 1 + e^{-K(z)}\bigr)\,dz \le \int_0^M K(z)\,dz.
\]
\end{lemma}
\begin{proof}
Let $f(t) = t - 1 + e^{-t}$ for $t \ge 0$, and note that $f'(t) = 1 - e^{-t} \ge 0$, so $f$ is non-decreasing. Let $B = K(M) \in [0, 1]$. We split the integral at $z = M$:
\begin{itemize}
    \item \textbf{On $[0, M]$:} Since $K(z)$ is non-increasing, $K(z) \ge K(M) = B$. Thus $1 - e^{-K(z)} \ge 1 - e^{-B}$, which implies $f(K(z)) \le K(z) - (1 - e^{-B})$. Integrating over $[0, M]$ gives
    \begin{equation}\label{eq:head}
        \int_0^M f(K(z))\,dz \le \int_0^M K(z)\,dz - M(1 - e^{-B}).
    \end{equation}
    \item \textbf{On $[M, \infty)$:} Since $K(z) \le \min(B, M/z)$ and $f$ is non-decreasing, we have
    \[
        \int_M^\infty f(K(z))\,dz \le \int_M^{M/B} f(B)\,dz + \int_{M/B}^\infty f\left(\frac{M}{z}\right) dz = M\left(\frac{1}{B} - 1\right)f(B) + M \int_0^B \frac{f(t)}{t^2}\,dt,
    \]
    where the second integral uses the substitution $t = M/z$. Integrating $\int_0^B f(t)t^{-2}\,dt$ by parts (noting $\lim_{t \to 0} f(t)/t = 0$) and using $f'(t) = 1 - e^{-t} \le t$ yields
    \[
        \int_0^B \frac{f(t)}{t^2}\,dt = \left[-\frac{f(t)}{t}\right]_0^B + \int_0^B \frac{1 - e^{-t}}{t}\,dt \le -\frac{f(B)}{B} + B.
    \]
    Substituting this back, the $f(B)/B$ terms cancel cleanly:
    \begin{equation}\label{eq:tail}
        \int_M^\infty f(K(z))\,dz \le M\left(\frac{f(B)}{B} - f(B) - \frac{f(B)}{B} + B\right) = M\bigl(B - f(B)\bigr) = M(1 - e^{-B}).
    \end{equation}
\end{itemize}
Summing~\eqref{eq:head} and~\eqref{eq:tail} cancels $M(1 - e^{-B})$ and proves the lemma.
\end{proof}

Define the truncated item values $V_j = \min(v_j, M)$ and their sum $V_M = \sum_{j=1}^n V_j$. Since $\mathbb{E}[V_j] = \int_0^M \Pr(v_j > z)\,dz = \int_0^M p_j(z)\,dz$, linearity of expectation gives
\[
    \int_0^M K(z)\,dz = \sum_{j=1}^n \mathbb{E}[V_j] = \mathbb{E}[V_M].
\]
Combining Lemmas~\ref{lem:cdw}, \ref{lem:survival}, and \ref{lem:truncate} therefore yields:
\begin{corollary}\label{cor:reduction}
$\OPT \le \SREV + \mathbb{E}[V_M]$.
\end{corollary}

\section{Bounding the Truncated Sum via Second Moments}\label{sec:second-moment}

To complete the proof, it remains to upper bound $E_M = \mathbb{E}[V_M]$ in terms of $M$. We do so by deriving both an upper bound and a lower bound on the second moment $S_2 = \mathbb{E}[V_M^2]$.

To derive the following upper bound, we fix the expectation $\mathbb{E}[V_M]$ and identify the \emph{extremal} distribution that maximizes the second moment $\mathbb{E}[V_M^2]$.
\begin{lemma}[Upper Bound on $S_2$]\label{lem:s2-upper}
Let $E_M = \mathbb{E}[V_M]$ and $S_2 = \mathbb{E}[V_M^2]$. Then
\[
    S_2 \le E_M^2 + M^2 y\left(\frac{E_M}{M}\right), \quad \text{where } y(x) = \begin{cases} x & \text{if } 0 \le x \le 1, \\ 2 - e^{1-x} & \text{if } x > 1. \end{cases}
\]
\end{lemma}
\begin{proof}
Let $S_j = \SRev_j$ and $E_j = \mathbb{E}[V_j]$. Let $G_j(z) = \Pr(V_j > z)$ be the survival function of $V_j \in [0, M]$. Our first goal is to upper bound the second moment of each item,
\[
    \mathbb{E}[V_j^2] = \int_0^M 2z \, G_j(z)\,dz,
\]
subject to three constraints on $G_j : [0, M] \to [0, 1]$:
\begin{enumerate}
    \item[(i)] \emph{Fixed mean (area):} $\int_0^M G_j(z)\,dz = E_j$,
    \item[(ii)] \emph{Single-item revenue cap (ceiling):} $G_j(z) \le \min(1, S_j/z)$ for all $z \in (0, M]$, and
    \item[(iii)] \emph{Monotonicity:} $G_j(z)$ is non-increasing in $z$.
\end{enumerate}
Because the weight $2z$ in $\int_0^M 2z \, G_j(z)\,dz$ is strictly increasing in $z$, maximizing $\mathbb{E}[V_j^2]$ for a fixed area $E_j$ requires pushing the area under $G_j(z)$ as far to the right (toward $M$) as constraints~(ii) and~(iii) allow. We analyze two cases depending on whether the ceiling $S_j/z$ binds:

\begin{itemize}
    \item \textbf{Case 1 ($E_j \le S_j$):} In this case, the ceiling $S_j/z$ never binds: a flat rectangle of constant height $E_j/M \le S_j/M \le S_j/z$ across $[0, M]$ already achieves area $E_j$. Pointwise, since $0 \le V_j \le M$, we immediately have $V_j^2 \le M V_j$, and taking expectations gives
    \[
        \mathbb{E}[V_j^2] \le M E_j = M S_j \left(\frac{E_j}{S_j}\right) = M S_j y\left(\frac{E_j}{S_j}\right).
    \]

    \item \textbf{Case 2 ($E_j > S_j$):} When $E_j > S_j$, we cannot shift all area to $z = M$ because the ceiling $G_j(z) \le S_j/z$ caps the height near $M$. Suppose we pack area tightly against the ceiling $S_j/z$ on some rightmost interval $[z_0, M]$. At $z = z_0$, the curve has height $S_j/z_0$, so the monotonicity constraint~(iii) forces $G_j(z) \ge S_j/z_0$ for all $z \in [0, z_0]$. To avoid wasting area on the left where the weight $2z$ is small, the extremal survival function $\tilde{G}_j$ stays as low as monotonicity allows on $[0, z_0]$—a flat plateau at height $S_j/z_0$—and hugs the ceiling $S_j/z$ on $(z_0, M]$:
    \[
        \tilde{G}_j(z) = \begin{cases} S_j / z_0 & \text{for } z \in [0, z_0], \\ S_j / z & \text{for } z \in (z_0, M]. \end{cases}
    \]
    Notice that the flat rectangle on $[0, z_0]$ always has area $z_0 (S_j/z_0) = S_j$, so the remaining area $E_j - S_j > 0$ must come from the tail $\int_{z_0}^M (S_j/z)\,dz = S_j \ln(M/z_0)$. Equating total area to $E_j$ uniquely determines the threshold $z_0$:
    \[
        \int_0^M \tilde{G}_j(z)\,dz = S_j + S_j \ln\left(\frac{M}{z_0}\right) = E_j \implies z_0 = M e^{1 - E_j/S_j}.
    \]
    (Since $E_j = \int_0^M G_j(z)\,dz \le \int_0^{S_j} 1\,dz + \int_{S_j}^M (S_j/z)\,dz = S_j(1 + \ln(M/S_j))$, we have $z_0 \in [S_j, M)$, so $\tilde{G}_j(z) \le 1$ is a valid survival function.)

    To prove formally that $\tilde{G}_j$ maximizes the second moment, we observe that $G_j$ and $\tilde{G}_j$ satisfy a single-crossing property:
    \begin{itemize}
        \item On $[0, z_0]$, $G_j(z)$ is non-increasing while $\tilde{G}_j(z) = S_j/z_0$ is constant, so $G_j(z) - \tilde{G}_j(z)$ is non-increasing and can cross zero from positive to negative at most once, at some point $z^* \in [0, z_0]$.
        \item On $(z_0, M]$, the ceiling constraint gives $G_j(z) \le S_j/z = \tilde{G}_j(z)$ everywhere.
    \end{itemize}
    Thus $G_j(z) \ge \tilde{G}_j(z)$ for $z < z^*$ and $G_j(z) \le \tilde{G}_j(z)$ for $z > z^*$, meaning that $\tilde{G}_j$ is obtained from $G_j$ by shifting area from left of $z^*$ to right of $z^*$ where the weight $2z$ is larger. Algebraically, $(z - z^*)(G_j(z) - \tilde{G}_j(z)) \le 0$ for all $z \in [0, M]$. Since $\int_0^M (G_j(z) - \tilde{G}_j(z))\,dz = E_j - E_j = 0$, we obtain
    \[
        \int_0^M 2z\bigl(G_j(z) - \tilde{G}_j(z)\bigr)\,dz = \int_0^M 2(z - z^*)\bigl(G_j(z) - \tilde{G}_j(z)\bigr)\,dz \le 0.
    \]
    Evaluating the second moment of $\tilde{G}_j$ directly gives the desired item bound:
    \[
        \mathbb{E}[V_j^2] \le \int_0^M 2z\,\tilde{G}_j(z)\,dz = \int_0^{z_0} 2z\left(\frac{S_j}{z_0}\right) dz + \int_{z_0}^M 2z\left(\frac{S_j}{z}\right) dz = S_j(2M - z_0) = M S_j y\left(\frac{E_j}{S_j}\right).
    \]
\end{itemize}

Finally, we aggregate across the $n$ items. Because $V_1, \dots, V_n$ are independent, their variances add: $\operatorname{Var}(V_M) = \sum_{j=1}^n \operatorname{Var}(V_j) \le \sum_{j=1}^n \mathbb{E}[V_j^2]$, and hence
\[
    S_2 = E_M^2 + \operatorname{Var}(V_M) \le E_M^2 + \sum_{j=1}^n M S_j y\left(\frac{E_j}{S_j}\right).
\]
Observe that $y$ is concave on $[0, \infty)$ with $y(0) = 0$. Defining non-negative weights $w_j = S_j/M$ for $j = 1, \dots, n$ and $w_0 = (M - \SREV)/M \ge 0$ (which sum to $1$), Jensen's inequality gives
\[
    \sum_{j=1}^n \frac{S_j}{M} y\left(\frac{E_j}{S_j}\right) = w_0 y(0) + \sum_{j=1}^n w_j y\left(\frac{E_j}{S_j}\right) \le y\left(w_0 \cdot 0 + \sum_{j=1}^n w_j \frac{E_j}{S_j}\right) = y\left(\frac{E_M}{M}\right).
\]
Multiplying by $M^2$ yields $S_2 \le E_M^2 + M^2 y(E_M/M)$.
\end{proof}

Symmetrically, to derive the following lower bound, we fix the
expectation $\mathbb{E}[V_M]$ and identify the extremal distribution
that \emph{minimizes} the second moment $\mathbb{E}[V_M^2]$.

\begin{lemma}[Lower Bound on $S_2$]\label{lem:s2-lower}
Let $E_M = \mathbb{E}[V_M]$ and $S_2 = \mathbb{E}[V_M^2]$.
If $E_M \ge M$, then
\[
  S_2 \;\ge\; M^2\bigl(2e^{E_M/M - 1} - 1\bigr),
  \quad\text{or equivalently,}\quad
  E_M \;\le\; M\!\left(1 + \ln\!\left(
      \frac{S_2 + M^2}{2M^2}\right)\right).
\]
\end{lemma}
\begin{proof}
Let $F(z)=\Pr(V_M>z)$ be the survival function of~$V_M$.
As before, $F$ satisfies the grand-bundle ceiling
\[
  0 \;\le\; F(z) \;\le\; \min\!\left(1,\,\frac{M}{z}\right)
  \quad\text{for all } z>0.
\]
For a fixed expectation $E_M = \int_0^\infty F(z)\,dz$,
we seek to minimize $S_2 = \int_0^\infty 2z\,F(z)\,dz$
subject to this ceiling.
Since placing a unit of area $F(z)\,dz$ at location~$z$
costs~$2z$ toward~$S_2$ but always contributes~$1$ toward~$E_M$,
the cheapest locations are near $z=0$.
Minimizing~$S_2$ for a fixed~$E_M$ therefore requires packing
area as far to the left as the ceiling allows---the exact
opposite of Lemma~\ref{lem:s2-upper}, which pushed area to the
right to maximize~$S_2$.

The extremal survival function saturates the ceiling from the
left up to a cutoff $z_1\ge M$ where the total area reaches~$E_M$:
\[
  F^*(z)=\begin{cases}
    1   & \text{for } 0\le z\le M,\\
    M/z & \text{for } M<z\le z_1,\\
    0   & \text{for } z>z_1,
  \end{cases}
  \qquad\text{where }
  z_1 = M e^{E_M/M-1},
\]
since
$\int_0^\infty F^*(z)\,dz = M + M\ln(z_1/M) = E_M$.
To verify that $F^*$ achieves the minimum, observe that
$F(z)\le F^*(z)$ on $[0,z_1]$
(where $F^*$ equals the ceiling) and
$F(z)\ge F^*(z)=0$ on $(z_1,\infty)$.
Hence $(z-z_1)\bigl(F(z)-F^*(z)\bigr)\ge 0$ everywhere.
Since both functions share the same expectation,
$\int_0^\infty(F-F^*)\,dz=0$, and so
\[
  \int_0^\infty 2z\bigl(F-F^*\bigr)\,dz
  \;=\;\int_0^\infty 2(z-z_1)\bigl(F-F^*\bigr)\,dz
  \;\ge\; 0.
\]
Evaluating the second moment of~$F^*$ directly gives the
desired bound:
\[
  S_2
  \;\ge\;\int_0^\infty 2z\,F^*(z)\,dz
  \;=\; M^2+2M(z_1-M)
  \;=\; M^2\bigl(2e^{E_M/M-1}-1\bigr).
  \qedhere
\]
\end{proof}

Combining the upper and lower bounds on~$S_2$ from
Lemmas~\ref{lem:s2-upper} and~\ref{lem:s2-lower} now allows us to upper bound $E_M$.
\begin{lemma}[Upper Bound on $E_M$]\label{lem:em-bound}
The expected truncated sum satisfies $E_M \le 2.5168\,M$.
\end{lemma}
\begin{proof}
If $E_M \le M$, the bound holds trivially.
Assume $x \coloneqq E_M/M > 1$.

\paragraph{Combining the two bounds.}
Dividing the equivalent form of Lemma~\ref{lem:s2-lower} by~$M$ gives
\[
  x \;\le\; 1 + \ln\!\left(\frac{S_2/M^2 + 1}{2}\right).
\]
Since $\ln$ is increasing, substituting the upper bound
$S_2/M^2 \le x^2 + 2 - e^{1-x}$ from Lemma~\ref{lem:s2-upper}
into the right-hand side yields the self-contained inequality
\[
  x \;\le\; g(x)
  \;\coloneqq\;
  1 + \ln\!\left(\frac{x^2 + 3 - e^{1-x}}{2}\right).
\]

\paragraph{Showing $g'(x) < 1$ for $x > 1$.}
Writing $g(x) = 1 + \ln u(x)$ with
$u(x) = \tfrac{1}{2}(x^2 + 3 - e^{1-x})$,
the chain rule gives
\[
  g'(x) = \frac{u'(x)}{u(x)}
  = \frac{2x + e^{1-x}}{x^2 + 3 - e^{1-x}},
\]
where we used $u'(x) = \tfrac{1}{2}(2x + e^{1-x})$.
Now $g'(x) < 1$ is equivalent to
\[
  2x + e^{1-x} \;<\; x^2 + 3 - e^{1-x},
  \qquad\text{i.e.,}\qquad
  2e^{1-x} \;<\; x^2 - 2x + 3 = (x-1)^2 + 2.
\]
For $x > 1$ the exponent $1-x$ is negative, so
$e^{1-x} < 1$ and therefore
$2e^{1-x} < 2 \le (x-1)^2 + 2$.
This confirms $g'(x) < 1$ for all $x > 1$.

\paragraph{Concluding.}
Define $h(x) \coloneqq x - g(x)$.
Since $h'(x) = 1 - g'(x) > 0$ for $x > 1$,
the function~$h$ is strictly increasing on $(1,\infty)$.
Any valid~$x$ must satisfy $x \le g(x)$, i.e.\ $h(x) \le 0$.
Evaluating at $x = 2.5168$:
\[
  g(2.5168)
  = 1 + \ln\!\left(\frac{2.5168^2 + 3 - e^{1-2.5168}}{2}\right)
  \approx 1 + \ln(4.5622)
  \approx 2.51676,
\]
so $h(2.5168) \approx 2.5168 - 2.51676 = 0.00004 > 0$.
Since $h$ is strictly increasing and already positive at
$x = 2.5168$, no $x \ge 2.5168$ can satisfy $h(x) \le 0$.
We conclude $x \le 2.5168$, i.e.\ $E_M \le 2.5168\,M$.
\end{proof}

\section{Proof of the Main Theorem}

\begin{theorem}\label{thm:main}
For a single additive buyer with $n$ independent items, the optimal expected revenue satisfies $\OPT \le 3.52 \max\{\SREV, \BREV\}$.
\end{theorem}
\begin{proof}
By Corollary~\ref{cor:reduction} and Lemma~\ref{lem:em-bound},
\[
    \OPT \le \SREV + \mathbb{E}[V_M] \le M + 2.5168 M = 3.5168 M \le 3.52 \max\{\SREV, \BREV\}. \qedhere
\]
\end{proof}

\bibliographystyle{alpha}
\bibliography{ref}
\end{document}